%% file: k-center-arxiv.tex
\documentclass[a4paper,11pt,UKenglish]{article}

\usepackage{bm,mathptmx,booktabs}
\usepackage{tcolorbox}
\usepackage[usestackEOL]{stackengine}

\usepackage[normalem]{ulem} %remove later..using it for strikeout
\usepackage{thm-restate}
\usepackage{amsmath,amssymb}

\DeclareMathAlphabet{\mathcal}{OMS}{cmsy}{m}{n}

\usepackage{tikz}
\usetikzlibrary{decorations.pathreplacing}

\input{preamble}

\date{}
\title{Tight Lower Bounds for Approximate \& Exact $k$-Center in $\mathbb{R}^d$\thanks{An extended abstract of this paper appears in SoCG 2022.}}

\author[1]{Rajesh~Chitnis%\thanks{Work done while at the University of Warwick, UK and supported by ERC grant 2014-CoG 647557.}
}
\author[2]{Nitin Saurabh}
\affil{School of Computer Science, University of Birmingham, UK. Email:
\texttt{rajeshchitnis@gmail.com}}
\affil[2]{Indian Institute of Technology Hyderabad, Sangareddy, India. \texttt{nitin@cse.iith.ac.in}}
\newcommand{\core}{\textsc{Core}}
\newcommand{\border}{\textsc{Border}}

\newcommand{\kcenter}{\textsc{Center}\xspace}
\newcommand{\kmedian}{\textsc{Median}\xspace}
\newcommand{\kmeans}{\textsc{Means}\xspace}
\newcommand{\OPT}{\text{OPT}\xspace}
\newcommand{\dist}{\texttt{dist}}

\newcommand{\geqcsp}{$\geq$-\text{CSP}\xspace}
\newcommand{\leqcsp}{$\leq$-\text{CSP}\xspace}
\newcommand{\bolda}{\mathbf{a}}
\newcommand{\ba}{\mathbf{a}}
\newcommand{\boldb}{\mathbf{b}}
\newcommand{\bb}{\mathbf{b}}
\newcommand{\boldaprime}{\mathbf{a'}}
\newcommand{\boldx}{\mathbf{x}}
\newcommand{\bx}{\mathbf{x}}
\newcommand{\by}{\mathbf{y}}

\newcommand{\bp}{\mathbf{p}}
\newcommand{\boldy}{\mathbf{y}}

\newcommand{\bq}{\mathbf{q}}

\newcommand{\bolde}{\mathbf{e}}
\newcommand{\be}{\mathbf{e}}
\newcommand{\R}{\textup{R}}

\newcommand{\VV}{\mathcal{V}}
\newcommand{\DD}{\mathcal{D}}
\newcommand{\CC}{\mathcal{C}}
\newcommand{\bone}{\mathbf{1}}

\usepackage{graphicx}
\usepackage{xspace,amssymb,amsmath,amsfonts}
\usepackage{multirow}

\usepackage{hyperref}
\usepackage{tikz,environ}
\usetikzlibrary{decorations.pathreplacing}
\usetikzlibrary{decorations.markings,arrows}
\tikzset{middlearrow/.style={
        decoration={markings,
            mark= at position 0.5 with {\arrow{#1}} ,
        },
        postaction={decorate}
    }
}

\usepackage[strings]{underscore}

\AtBeginDocument{\def\sectionautorefname{Section}}
\AtBeginDocument{\def\subsectionautorefname{Section}}
\AtBeginDocument{\def\subsubsectionautorefname{Section}}
\definecolor{darkblue}{rgb}{0,0,1}
\definecolor{darkred}{rgb}{0.6,0,0}
\definecolor{darkgreen}{rgb}{0,1,0}
\hypersetup{colorlinks, linkcolor=darkblue, citecolor=darkblue,
urlcolor=darkblue}

\begin{document}
\renewcommand*{\sectionautorefname}{Section}
\renewcommand*{\subsectionautorefname}{Section}
\renewcommand*{\subsubsectionautorefname}{Section}

\maketitle

\begin{abstract}
In the discrete $k$-\kcenter problem, we are given a metric space $(P,\dist)$ where $|P|=n$ and the goal is to select a set $C\subseteq P$ of $k$ centers which minimizes the maximum distance of a point in $P$ from its nearest center. For any $\epsilon>0$, Agarwal and Procopiuc [SODA '98, Algorithmica '02] designed an $(1+\epsilon)$-approximation algorithm\footnote{The algorithm of Agarwal and Procopiuc~\cite{agarwal-procopiuc} also works for the non-discrete, i.e., continuous, version of the problem when $C$ need not be a subset of $P$, but our lower bounds only hold for the discrete version.} for this problem in $d$-dimensional Euclidean space\footnote{The algorithm of Agarwal and Procopiuc~\cite{agarwal-procopiuc} also works for other metrics such as $\ell_{\infty}$ or $\ell_q$ metric for $q\geq 1$. Our construction also works for $\ell_{\infty}$ (in fact, some of the bounds are simpler to derive!) but we present only the proof for $\ell_2$ to keep the presentation simple.} which runs in $O(dn\log k) + \left(\dfrac{k}{\epsilon}\right)^{O\left(k^{1-1/d}\right)}\cdot n^{O(1)}$ time. In this paper we show that their algorithm is essentially optimal: if for some $d\geq 2$ and some computable function $f$, there is an $f(k)\cdot \left(\dfrac{1}{\epsilon}\right)^{o\left(k^{1-1/d}\right)} \cdot n^{o\left(k^{1-1/d}\right)}$ time algorithm for $(1+\epsilon)$-approximating the discrete $k$-\kcenter on $n$ points in $d$-dimensional Euclidean space then the Exponential Time Hypothesis (ETH) fails.

We obtain our lower bound by designing a gap reduction from a $d$-dimensional constraint satisfaction problem (CSP) to discrete $d$-dimensional $k$-\kcenter. This reduction has the property that there is a fixed value $\epsilon$ (depending on the CSP) such that the optimal radius of $k$-\kcenter instances corresponding to satisfiable and unsatisfiable instances of the CSP is $<1$ and $\geq (1+\epsilon)$ respectively. Our claimed lower bound on the running time for approximating discrete $k$-\kcenter in $d$-dimensions then follows from the lower bound due to Marx and Sidiropoulos [SoCG '14] for checking the satisfiability of the aforementioned $d$-dimensional CSP.

As a byproduct of our reduction, we also obtain that the exact algorithm of Agarwal and Procopiuc [SODA '98, Algorithmica '02] which runs in $n^{O\left(d\cdot k^{1-1/d}\right)}$ time for discrete $k$-\kcenter on $n$ points in $d$-dimensional Euclidean space is asymptotically optimal. Formally, we show that if for some $d\geq 2$ and some computable function $f$, there is an $f(k)\cdot n^{o\left(k^{1-1/d}\right)}$ time exact algorithm for the discrete $k$-\kcenter problem on $n$ points in $d$-dimensional Euclidean space then the Exponential Time Hypothesis (ETH) fails.  Previously, such a lower bound was only known for $d=2$ and was implicit in the work of Marx [IWPEC '06].

\end{abstract}

\section{Introduction}

The $k$-\kcenter problem is a classical problem in theoretical computer science and was first formulated by Hakimi~\cite{hakimi} in 1964. In this problem, given a metric space $(P,\dist)$ and an integer $k\leq |P|$ the goal is to select a set $C$ of $k$ centers which minimizes the maximum distance of a point in $P$ from its nearest center, i.e., select a set $C$ which minimizes the quantity $\max_{p\in P} \min_{c\in C} \dist(p,c)$. A geometric way to view the $k$-\kcenter problem is to find the minimum radius $r$ such that $k$ closed balls of radius $r$ located at each of the points in $C$ cover all the points in $P$. In most applications, we require that $C\subseteq P$ and this is known as the discrete version of the problem.

As an example, one can consider the set $P$ to be important locations in a city and solving the $k$-\kcenter problem (where $k$ is upper bounded by budget constraints) establishes the locations of fire stations which minimize the response time in event of a fire. In addition to other applications in facility location, transportation networks, etc. an important application of $k$-\kcenter is in clustering. With the advent of massive data sets, the problem of efficiently and effectively summarizing this data is crucial. A standard approach for this is via centroid-based clustering algorithms of which $k$-\kcenter is a special case. Clustering using $k$-\kcenter has found applications in text summarization, robotics, bioinformatics, pattern recognition, etc ~\cite{clustering-text,clustering-robotics,hennig2015handbook,jiang2004cluster}.%\\

\subsection{Prior work on exact \& approximate algorithms for discrete $k$-\kcenter}

The discrete\footnote{Here we mention the known results only for the discrete version of $k$-\kcenter. A discussion about results for the continuous version of the problem is given in~\autoref{subsec:continuous-k-center-discussion}.} $k$-\kcenter problem is NP-hard~\cite{vazirani-book}, and admits a $2$-approximation~\cite{DBLP:journals/mor/HochbaumS85,DBLP:journals/tcs/Gonzalez85} in $n^{O(1)}$ time where $n$ is the number of points. This approximation ratio is tight and the $k$-\kcenter problem is NP-hard to approximate in polynomial time to a factor $(2-\epsilon)$ for any constant $\epsilon>0$~\cite{DBLP:journals/dam/HsuN79,DBLP:journals/tcs/Gonzalez85}. Given this intractability, research was aimed at designing parameterized algorithms~\cite{fpt-book} and parameterized approximation algorithms for $k$-center. The $k$-\kcenter problem is W[2]-hard to approximate to factor better than $2$ even when allowing running times of the form $f(k)\cdot n^{O(1)}$ for any computable function $f$~\cite{DBLP:journals/algorithmica/Feldmann19,DBLP:journals/talg/DemaineFHT05}. The $k$-\kcenter problem remains W[2]-hard even if we combine the parameter $k$ with other structural parameters such as size of vertex cover or size of feedback vertex set~\cite{DBLP:journals/dam/KatsikarelisLP19}. Agarwal and Procopiuc~\cite{agarwal-procopiuc} designed an algorithm for discrete $k$-\kcenter on $n$ points in $d$-dimensional Euclidean space which runs in $n^{O\left(d\cdot k^{1-1/d}\right)}$ time.

The paradigm of combining parameterized algorithms \& approximation algorithms has been successful in designing algorithms for $k$-center in special topologies such as $d$-dimensional Euclidean space~\cite{agarwal-procopiuc}, planar graphs~\cite{DBLP:conf/soda/Fox-EpsteinKS19}, metrics of bounded doubling dimensions~\cite{DBLP:journals/algorithmica/FeldmannM20}, graphs of bounded highway dimension~\cite{DBLP:journals/algorithmica/Feldmann19,DBLP:conf/esa/BeckerKS18}, etc. Of particular relevance to this paper is the $(1+\epsilon)$-approximation algorithm\footnote{This is also known as an efficient parameterized approximation scheme (EPAS) as the running time is a function of the type $f(k,\epsilon,d)\cdot n^{O(1)}$.} of Agarwal and Procopiuc~\cite{agarwal-procopiuc} which runs in $O(dn\log k) + \left(\dfrac{k}{\epsilon}\right)^{O\left(k^{1-1/d}\right)}\cdot n^{O(1)}$ time. This was generalized by Feldmann and Marx~\cite{DBLP:journals/algorithmica/FeldmannM20} who designed an $(1+\epsilon)$-approximation algorithm running in $\left(\dfrac{k^k}{\epsilon^{O(kD)}}\right)\cdot n^{O(1)}$ time for discrete $k$-\kcenter in metric spaces of doubling dimension $D$.

\subsection{From 2-dimensions to higher dimensions}

\subparagraph*{Square root phenomenon for planar graphs and geometric problems in the plane:} For a wide range of problems on planar graphs or geometric problems in the plane, a certain {\em square root phenomenon} is observed for a wide range of algorithmic problems: the exponent of the running time can be improved from $O(\ell)$ to $O(\sqrt{\ell})$ where $\ell$ is the parameter, or from $O(n)$ to $O(\sqrt{n}$) where $n$ is in the input size, and lower bounds indicate that this improvement is essentially best possible. There is an ever increasing list of such problems known for planar graphs~\cite{DBLP:journals/siamcomp/ChitnisFHM20,DBLP:conf/icalp/Marx12,DBLP:conf/icalp/KleinM12,MarxPP-FOCS2018,KleinM14,DemaineFHT05,DBLP:conf/stacs/PilipczukPSL13,DBLP:conf/esa/MarxP15,DBLP:conf/fsttcs/LokshtanovSW12,DBLP:journals/corr/AboulkerBHMT15,FominLMPPS16}
and in the plane~\cite{DBLP:conf/esa/MarxP15,DBLP:conf/iwpec/Marx06,FominKLPS16,DBLP:journals/jal/AlberF04,DBLP:conf/focs/SmithW98,DBLP:journals/algorithmica/HwangLC93,DBLP:journals/algorithmica/HwangCL93}

\subparagraph*{Bounds for higher dimensional Euclidean spaces:} Unlike the situation on planar graphs and in two-dimensions, the program of obtaining tight bounds for higher dimensions is still quite nascent with relatively fewer results~\cite{bane,DBLP:conf/compgeom/MarxS14,biro-higher-d,de-berg-higher-d,tsp-higher-d}. Marx and Sidiropoulos~\cite{DBLP:conf/compgeom/MarxS14} showed that for some problems there is a \emph{limited blessing of low dimensionality}: \textcolor{black}{that is,} for $d$-dimensions the running time can be improved from $n^{\ell}$ to $n^{\ell^{1-1/d}}$ or from $2^{n}$ to $2^{n^{1-1/d}}$ where $\ell$ is a parameter and $n$ is the input size. In contrast, Cohen-Addad et al.~\cite{bane} showed that the two problems of $k$-\kmedian and $k$-\kmeans suffer from the \emph{curse of low dimensionality}: even for $4$-dimensional Euclidean space, assuming the Exponential Time Hypothesis\footnote{Recall that the Exponential Time Hypothesis (ETH) has the consequence that $n$-variable 3-SAT cannot be solved in $2^{o(n)}$ time~\cite{eth,eth-2}.} (ETH), there is no $f(k)\cdot n^{o(k)}$ time algorithm, i.e., the brute force algorithm which runs in $n^{O(k)}$ time is asymptotically optimal.

\subsection{Motivation \& Our Results}

In two-dimensional Euclidean space there is an $n^{O(\sqrt{k})}$ algorithm~\cite{agarwal-procopiuc,DBLP:journals/algorithmica/HwangLC93,DBLP:journals/algorithmica/HwangCL93}, and a matching lower bound of $f(k)\cdot n^{o(\sqrt{k})}$ under Exponential Time Hypothesis (ETH) for any computable function $f$~\cite{DBLP:conf/iwpec/Marx06}. Our motivation in this paper is to investigate what is the \emph{correct} complexity of exact and approximate algorithms for the discrete $k$-\kcenter for higher dimensional Euclidean spaces. In particular, we aim to answer the following two questions:
\begin{table}[ht]
\noindent\framebox{\begin{minipage}{\textwidth}
\begin{description}
  \item[(Question 1)] Can the running time of the $(1+\epsilon)$-approximation algorithm of \cite{agarwal-procopiuc} be improved from $O(dn\log k) + \left(\dfrac{k}{\epsilon}\right)^{O\left(k^{1-1/d}\right)}\cdot n^{O(1)}$, or is there a (close to) matching lower bound?
  \item[(Question 2)] The $n^{O\left(d\cdot k^{1-1/d}\right)}$ algorithm of \cite{agarwal-procopiuc} for $d$-dimensional Euclidean space shows that there is a \emph{limited blessing of low dimensionality} for $k$-\kcenter. But can the term $k^{1-1/d}$ in the exponent be improved, or is it asymptotically tight?
\end{description}
\end{minipage}}
\end{table}

\noindent We make progress towards answering both these questions by showing the following theorem:

\begin{restatable}{theorem}{domsetd}
\normalfont
For any $d\geq 2$, under the Exponential Time Hypothesis (ETH), the discrete $k$-\kcenter problem in $d$-dimensional Euclidean space
\begin{description}

  \item[\textbf{- (Inapproximability result)}] does not admit an $(1+\epsilon)$-approximation in $f(k)\cdot \left(\frac{1}{\epsilon}\right)^{o\left(k^{1-1/d}\right)}\cdot n^{o\left(k^{1-1/d}\right)}$ time where $f$ is any computable function and $n$ is the number of points.

  \item[\textbf{- (Lower bound for exact algorithm)}] cannot be solved in $f(k)\cdot n^{o\left(k^{1-1/d}\right)}$ time where $f$ is any computable function and $n$ is the number of points.
\end{description}
\label{thm:dom-set-d-dimensions}
\end{restatable}
%\vspace{3mm}

\autoref{thm:dom-set-d-dimensions} answers Question~$1$ by showing that the running time of the $(1+\epsilon)$-approximation algorithm of Agarwal and Procopiuc \cite{agarwal-procopiuc} is essentially tight, i.e., the dependence on $\epsilon$ cannot be improved even if we allow a larger dependence on both $k$ and $n$. \autoref{thm:dom-set-d-dimensions} answers Question~$2$ by showing that the running time of the exact algorithm of Agarwal and Procopiuc \cite{agarwal-procopiuc} is asymptotically tight, i.e., the exponent of $k^{1-1/d}$ cannot be asymptotically improved even if we allow a larger dependence on $k$.

\subsection{Discussion of the continuous $k$-\kcenter problem}
\label{subsec:continuous-k-center-discussion}
In the continuous version of the $k$-\kcenter problem, the centers are not required to be picked from the original set of input points. The $n^{O\left(d\cdot k^{1-1/d}\right)}$ algorithm of Agarwal and Procopiuc~\cite{agarwal-procopiuc} also works for this continuous version of the $k$-\kcenter problem in $\mathbb{R}^d$. Marx \cite{dm-esa-05} showed the W[1]-hardness of $k$-\kcenter in $(\mathbb{R}^2, \ell_{\infty})$ parameterized by $k$. Cabello et al.~\cite{dm-continuous} studied the complexity of this problem parameterized by the dimension, and showed the W[1]-hardness of $4$-\kcenter in $(\mathbb{R}^d, \ell_{\infty})$ parameterized by $d$. Additionally, they also obtained the W[1]-hardness of $2$-\kcenter in $(\mathbb{R}^d, \ell_{2})$ parameterized by $d$\textcolor{black}{;} this reduction also rules out existence of $n^{o(d)}$ algorithms for this problem under the Exponential Time Hypothesis (ETH). It is an interesting open question whether the $n^{O\left(d\cdot k^{1-1/d}\right)}$ algorithm of Agarwal and Procopiuc \cite{agarwal-procopiuc} is also asymptotically tight for the continuous version of the problem: one way to possibly prove this would be to extend the W[1]-hardness reduction of Marx \cite{dm-esa-05} for continuous $k$-\kcenter in $\mathbb{R}^2$ (parameterized by $k$) to higher dimensions using the framework of Marx and Sidiropoulos \cite{DBLP:conf/compgeom/MarxS14}. Our reduction in this paper does not extend to the continuous version.

\subsection{Notation}
The set $\{1,2,\ldots, n\}$ is denoted by $[n]$. All vectors considered in this paper have length $d$. If $\ba$ is a vector then for each $i\in [d]$ its $i$\textcolor{black}{-th}  %$i^{\text{th}}$
coordinate is denoted by $\ba[i]$. Addition and subtraction of vectors is denoted by $\oplus$ and $\ominus$ respectively. The $i$\textcolor{black}{-th} %$i^{\text{th}}$
unit vector is denoted by $\bolde_i$ and has $\be_{i}[i]=1$ and $\be_{i}[j]=0$ for each $j\neq i$. The $d$-dimensional vector \textcolor{black}{whose every} %which has all
coordinate \textcolor{black}{equals} %as
$1$ is denoted by $\bone^d$. If $u$ is a point and $X$ is a set of points then $\dist(u, X) = \min_{x\in X} \dist(u,x)$. We will sometimes abuse notation slightly and use $x$ to denote both the name and location of the point $x$.

\section{Lower bounds for exact \& approximate $k$-\kcenter in $d$-dimensional Euclidean space}
\label{sec:general-d}

The goal of this section is to prove \autoref{thm:dom-set-d-dimensions} which is restated below:

\domsetd*

\subparagraph*{Roadmap to prove \autoref{thm:dom-set-d-dimensions}:} To prove \autoref{thm:dom-set-d-dimensions}, we design a gap reduction (described in \autoref{subsec:redn-general-d}) from a constraint satisfaction problem (CSP) to the $k$-\kcenter problem. The definition and statement of the lower bound for the CSP due to Marx and Sidiropoulos \cite{DBLP:conf/compgeom/MarxS14} is given in \autoref{subsec:marx-sidiropoulos}. The correctness of the reduction is shown in~\autoref{subsec:k-center-general-d-easy} and~\autoref{subsec:k-center-general-d-hard}. Finally, everything is tied together in \autoref{subsec:finishing-the-proof} which contains the proof of \autoref{thm:dom-set-d-dimensions}.

\subsection{Lower bound for $d$-dimensional geometric \geqcsp~\cite{DBLP:conf/compgeom/MarxS14}}
\label{subsec:marx-sidiropoulos}

This section introduces the $d$-dimensional geometric \geqcsp problem of Marx and Sidiropoulos \cite{DBLP:conf/compgeom/MarxS14}. First we start with some definitions before stating the formal lower bound (\autoref{thm:marx-sidiropoulos}) that will be used to prove \autoref{thm:dom-set-d-dimensions}. Constraint Satisfaction Problems (CSPs) are a general way to represent several important problems in theoretical computer science. In this paper, we will only need a subclass of CSPs called binary CSPs which we define below.
\begin{definition}
\normalfont
An instance of a binary constraint satisfaction problem (CSP) is a triple $\mathcal{I}=(\VV, \DD, \CC)$ where $\VV$ is a set of variables, $\DD$ is a domain of values and $\CC$ is a set of constraints. There are two types of constraints:
\begin{itemize}%[leftmargin=*,align=left]
  \item \underline{\emph{Unary constraints}}: For some $v\in \VV$ there is a unary constraint $\langle v, R_v \rangle$ where $R_v \subseteq \DD$.
  \item \underline{\emph{Binary constraints}}: For some $u,v\in \VV$, \textcolor{black}{$u \neq v$}, there is a binary constraint $\big\langle (u,v), R_{u,v} \big\rangle$ where $R_{u,v}\subseteq \DD \times \DD$.
\end{itemize}
\label{defn:binary-csp}
\end{definition}

Solving a given CSP instance $\mathcal{I}=(\VV,\DD,\CC)$ is to check whether there exists a satisfying assignment for it, i.e., a function $f:\VV\to \DD$ such that all the constraints are satisfied. For a binary CSP, a satisfying assignment $f$ has the property that for each unary constraint $\langle v, R_v \rangle$ we have $f(v)\in R_v$ and for each binary constraint $\big\langle (u,v), R_{u,v} \big\rangle$ we have $\left(f(u),f(v)\right)\in R_{u,v}$.

The constraint graph of a given CSP instance $\mathcal{I}=(V,D,C)$ is an undirected graph $G_{\mathcal{I}}$ whose vertex set is $V$ and the adjacency relation is defined as follows: two vertices $u,v\in V$ are adjacent in $G_{\mathcal{I}}$ if there is a constraint in $\mathcal{I}$ which contains both $u$ and $v$.
Marx and Sidiropoulos~\cite{DBLP:conf/compgeom/MarxS14} observed that
binary CSPs whose primal graph is a subgraph of the $d$-dimensional grid are useful in showing lower bounds for geometric problems in $d$-dimensions.

\begin{definition}
\normalfont
The $d$-dimensional grid $\R[N, d]$ is an undirected graph with vertex set $[N]^d$ and the adjacency relation is as follows: two vertices $(a_1, a_2, \ldots, a_d)$ and $(b_1, b_2, \ldots, b_d)$ have an edge between them if and only if $\sum_{i=1}^{d} |a_i-b_i|=1$.
\label{defn:d-dimensional-grid}
\end{definition}

\begin{definition}
\normalfont
A $d$-dimensional geometric \geqcsp $\mathcal{I}=(\VV,\DD,\CC)$ is a binary CSP whose
\begin{itemize}
  \item set of variables $\VV$ is a subset of $\R[N,d]$ for some $N\geq 1$,
  \item domain is $[\delta]^d$ for some integer $\delta\geq 1$,
  \item constraint graph $G_{\mathcal{I}}$ is an \emph{induced} subgraph of $\R[N,d]$,
  \item unary constraints are arbitrary, \textcolor{black}{and}
  \item binary constraints are of the following type: if $\bolda,\bolda' \in \VV$ such that $\bolda'=\bolda\oplus\bolde_i$ for some $i\in [d]$ then there is a binary constraint $\big\langle (\bolda,\bolda'), R_{\bolda,\bolda'} \big\rangle$ where
      $R_{\bolda,\bolda'}=\left\{ \left( \bx,\by\right)\in R_{\ba}\times R_{\ba'}\ \mid\ \bx[i]\geq \by[i] \right\}$.
\end{itemize}
\label{defn:d-dimensional-geometric-leqcsp}
\end{definition}

Observe that the set of unary constraints of a $d$-dimensional geometric \geqcsp is sufficient to completely define it. The size $|\mathcal{I}|$ of a binary CSP $\mathcal{I}=(\VV,\DD,\CC)$ is the combined size of the variables, domain and the constraints. With appropriate preprocessing (e.g., combining different constraints on the same variables) we can assume that $|\mathcal{I}|=\left(|\VV|+|\DD|\right)^{O(1)}$. We now state the result of Marx and Sidiropoulos \cite{DBLP:conf/compgeom/MarxS14} which gives a lower bound on the complexity of checking whether a given $d$-dimensional geometric \geqcsp has a satisfying assignment.

\begin{theorem}
\normalfont
\citep[Theorem 2.10]{DBLP:conf/compgeom/MarxS14} If for some fixed $d \geq 2$, there is an $f(|\VV|)\cdot |\mathcal{I}|^{o\left(|\VV|^{1-1/d}\right)}$ time algorithm for solving a $d$-dimensional geometric \geqcsp $\mathcal{I}$ for some computable function $f$, then the Exponential Time Hypothesis (ETH) fails.
\label{thm:marx-sidiropoulos}
\end{theorem}

\begin{remark}
\label{remark:issue-of-geqcsp-versus-leqcsp}
\normalfont
The problem defined by Marx and Sidiropoulos \cite{DBLP:conf/compgeom/MarxS14} is actually $d$-dimensional geometric \leqcsp which has $\leq$-constraints instead of the $\geq$-constraints. However, for each $\bolda\in \VV$ by replacing each unary constraint $\bx\in R_{\bolda}$ by $\by$ such that $\by[i]=N+1-\bx[i]$ for each $i\in [d]$, it is easy to see that $d$-dimensional geometric \leqcsp and $d$-dimensional geometric \geqcsp are equivalent.
\end{remark}

\subsection{Reduction from $d$-dimensional geometric \geqcsp to $k$-\kcenter in $\mathbb{R}^d$}
\label{subsec:redn-general-d}

We are now ready to describe our reduction from $d$-dimensional geometric \geqcsp to $k$-\kcenter in $\mathbb{R}^d$.
Fix any $d\geq 2$.
Let $\mathcal{I}=(\VV,\DD,\CC)$ be a $d$-dimensional geometric \geqcsp instance on variables $\VV$ and domain $[\delta]^d$ for some integer $\delta\geq 1$. We fix\footnote{For simplicity of presentation, we choose $r=1/4$ instead of $r=1$: by scaling the result holds also for $r=1$.} the following two quantities:
\begin{gather}\label{defn:r-epsilon-general-d}
r:=\frac{1}{4}\quad \text{and}\quad \textcolor{black}{\epsilon:=\frac{r^2}{(d-1)\delta^{2}} = \frac{1}{16(d-1)\delta^2}}.
\end{gather}
\textcolor{black}{Since $d\geq 2$ and $\delta\geq 1$, we obtain the following bounds from} \autoref{defn:r-epsilon-general-d}, %(which will be used throughout the remainder of the paper):
\begin{align}
  0 < \epsilon\leq \epsilon\delta\leq \epsilon\delta^2\leq \epsilon\delta^2(d-1)= r^2 = \frac{1}{16}.  \label{eqn:always-to-be-cited}
\end{align}

Given an instance $\mathcal{I}=(\VV, \DD, \CC)$ of $d$-dimensional geometric \geqcsp, we add a set $\mathcal{U}$ of points in $\mathbb{R}^d$ as described in \autoref{table:general-d:construction} \textcolor{black}{and} \autoref{table-special-sets-general-d}. These set of points are the input for the instance of the $|\VV|$-\kcenter problem.
\begin{table}[ht]
\noindent\framebox{\begin{minipage}{\textwidth}
\begin{enumerate}
  \item[(1)] \underline{Corresponding to variables}: If $\bolda\in \VV$ then we add the following set of points which are collectively called as $\border[\bolda]$
    \begin{itemize}
    \item For each $i\in [d]$, the point $B_{\bolda}^{+i}$ which is located at $\bolda \oplus \bolde_{i}\cdot \textcolor{black}{r(1-\epsilon)}\oplus (\bone^d-\be_i)\cdot 2\epsilon \delta$.
    \item For each $i\in [d]$, the point $B_{\bolda}^{-i}$ which is located at $\bolda \ominus \bolde_{i}\cdot \textcolor{black}{r(1-\epsilon)} \ominus (\bone^d-\be_i)\cdot 2\epsilon \delta$.
    \end{itemize}
    This set of points are \textcolor{black}{referred to} as \emph{border} points.
  \item[(2)] \underline{Corresponding to unary constraints}: If $\bolda\in \VV$ and $\big\langle (\bolda), R_{\bolda} \big\rangle$ is the unary constraint on $\bolda$, then we add the following set of points which are collectively called as $\core[\bolda]$:
    \begin{itemize}
    \item  for each $\boldx\in R_{\bolda}\subseteq [\delta]^d$ we add a point called $C_{\bolda}^{\boldx}$ located at $\bolda \oplus \epsilon\cdot \boldx$.
    \end{itemize}
    This set of points are \textcolor{black}{referred to} as \emph{core} points.
  \item[(3)] \underline{Corresponding to adjacencies in $G_{\mathcal{I}}$}: For every edge $(\bolda, \boldaprime)$ in $G_{\mathcal{I}}$ we add a collection of $\delta$ points denoted by $\mathcal{S}_{\{\bolda,\bolda'\}}$. Assume, without loss of generality, that $\boldaprime = \bolda\oplus\bolde_i$ for some $i\in[d]$.

    Then the set of points $\mathcal{S}_{\{\bolda,\bolda'\}}$ is defined as follows:
    \begin{itemize}
    \item for each $\ell\in [\delta]$ we add a point $S_{\{\bolda,\bolda'\}}^{\ell}$ which is located at $\bolda\oplus\bolde_{i}\cdot\left( (1-\epsilon)2r+\epsilon\ell\right)$.
    \end{itemize}
    This set of points are \textcolor{black}{referred to} as \emph{secondary} points.
\end{enumerate}
\end{minipage}}
\vspace{2mm}
\caption{The set $\mathcal{U}$ of points in $\mathbb{R}^d$ \big(which gives an instance of $k$-\kcenter\big) constructed from an instance $\mathcal{I}=(\VV,\DD,\CC)$ of $d$-dimensional geometric \geqcsp.}
\label{table:general-d:construction}
\end{table}

\textcolor{black}{Note that we add at most $|\VV|\cdot 2d$ many border points, at most $|\CC|$ many core points, and at most $|\VV|^{2}\cdot \delta$ many secondary points}. Hence, the total number of points $n$ in the instance $\mathcal{U}$ is $\leq |\VV|\cdot 2d + |\CC| +|\VV|^{2}\cdot \delta = |\mathcal{I}|^{O(1)}$ where $|\mathcal{I}|=|\VV|+|\DD|+|\CC|$. We now prove some preliminary lemmas to be later used in~\autoref{subsec:k-center-general-d-easy} and~\autoref{subsec:k-center-general-d-hard}.

\begin{table}[ht]
\noindent\framebox{\begin{minipage}{\textwidth}
\begin{gather}
\text{For each}\ \bolda\in \VV,\ \text{let}\ \mathcal{D}[\bolda] := \core[\bolda]\ \bigcup\ \border[\bolda]. \label{eqn:defn-of-d[a]-general-d} \\
\text{The set of primary points is \textsc{Primary}}:= \bigcup_{\bolda\in \VV} \mathcal{D}[\bolda]. \label{eqn:defn-of-primary-balls-general-d} \\
\text{The set of secondary points is \textsc{Secondary}} := \bigcup_{\bolda\ \&\ \boldaprime\ \text{forms an edge in}\ G_\mathcal{I}} \mathcal{S}_{\{\bolda,\bolda'\}}. \label{eqn:defn-of-secondary-balls-general-d}\\
\text{The final collection of points is}\ \mathcal{U} := \textsc{Primary}\ \bigcup\ \textsc{Secondary}. \label{eqn:defn-of-set-of-ball-general-d}
\end{gather}
\end{minipage}}
\vspace{2mm}
\caption{Notation for some special subsets of points from $\mathcal{U}$. Note that a primary point is either a core point or a border point.}
\label{table-special-sets-general-d}
\end{table}

\subsubsection{Preliminary lemmas}

\begin{lemma}
\normalfont
For each $\bolda\in \VV$ and $i\in [d]$, we have $\dist\left( B_{\ba}^{+i}, B_{\ba}^{-i}\right)\geq 2r(1+\epsilon)$.
\label{lem:borders-dont-intersect-2d}
\end{lemma}
\begin{proof}
Fix any $\bolda\in \VV$ and $i\in [d]$. By~\autoref{table:general-d:construction}, the points $B_{\bolda}^{+i}$ and $B_{\bolda}^{-i}$ are located at $\bolda \oplus \bolde_{i}\cdot r(1-\epsilon) \oplus (\bone^d-\be_i)\cdot 2\epsilon \delta$ and $\bolda \ominus \bolde_{i}\cdot r(1-\epsilon) \ominus (\bone^d-\be_i)\cdot 2\epsilon \delta$ respectively. Hence, we have that
\begin{align*}
  \dist\left(B_{\ba}^{+i}, B_{\ba}^{-i}\right)^2
  & = (2r(1-\epsilon))^{2} + (d-1)\cdot (4\epsilon \delta)^{2} = (2r(1-\epsilon))^{2} + 16\epsilon \cdot (d-1)\epsilon\delta^{2}, \\
  & = (2r(1-\epsilon))^{2} + 16\epsilon \cdot r^2, \tag{by definition of $\epsilon$ in \autoref{defn:r-epsilon-general-d}}\\
  & = (2r)^2[(1-\epsilon)^2 + 4\epsilon] = (2r(1+\epsilon))^2.
\end{align*}
\end{proof}

\begin{lemma}
\normalfont
For each $\bolda\in \VV$, the distance between any two points in $\core[\ba]$ is $< r$.
\label{lem:core-pairwise-intersects-general-d}
\end{lemma}
\begin{proof}
Fix any $\ba \in \VV$. Consider any two points in $\core[\ba]$, say $C_{\ba}^{\bx}$ and $C_{\ba}^{\by}$, for some $\bx\neq \by$. By \autoref{table:general-d:construction}, these points are located at $\ba\oplus \epsilon\cdot \bx$ and $\ba\oplus \epsilon\cdot \by$ respectively. Hence, we have
\begin{align*}
  \dist\left( C_{\ba}^{\bx}, C_{\ba}^{\by}\right)^2  & = \left(\epsilon\cdot \dist(\bx,\by)\right)^{2}, \\
  & \leq \epsilon^2\cdot d\cdot (\delta-1)^{2}, \tag{since $\bx,\by \in R_{\ba} \subseteq [\delta]^d$}\\
  & = \frac{d(\delta-1)^2}{(d-1)^2\delta^4}\cdot r^4, \tag{by definition of $\epsilon$ in \autoref{defn:r-epsilon-general-d}}\\
  & \leq \frac{1}{8}\cdot r^4 < r. \tag{since $d\geq 2$ and $\delta \geq 1$}
\end{align*}
\end{proof}

\begin{lemma}
\normalfont
For each $\ba\in \VV$, the distance of any point from $\core[\ba]$ to any point from $\border[\ba]$ is $<2r$.
\label{lem:core-intersects-all-border-general-d}
\end{lemma}
\begin{proof}
  Fix any $\bolda\in \VV$ and consider any point $C_{\bolda}^{\boldx}\in \core[\bolda]$ where $\boldx\in R_{\bolda}\subseteq [\delta]^{d}$.
We prove this lemma by showing that, for each $i\in [d]$, the point $C_{\bolda}^{\boldx}$ is \textcolor{black}{at} distance $< 2r$ from \textcolor{black}{both} the points $B_{\bolda}^{+i}$ and $B_{\bolda}^{-i}$. Fix some $i\in [d]$.
\begin{enumerate}
\item[(i)] By \autoref{table:general-d:construction}, the points $C_{\bolda}^{\boldx}$ and $B_{\bolda}^{+i}$ are located at $\bolda\oplus \epsilon\cdot \boldx$ and  $\bolda \oplus \bolde_{i}\cdot r(1-\epsilon) \oplus (\bone^d-\be_i)\cdot 2\epsilon \delta$ respectively. Hence, we have
  \begin{align*}
    \dist\left(C_{\ba}^{\bx} , B_{\bolda}^{+i}\right)^2 &
    = (r(1-\epsilon)-\epsilon\cdot \bx[i])^{2} + \sum_{j=1 \colon j\neq i}^{d} (2\epsilon \delta- \epsilon\cdot \bx[j] )^2, \\
    & \leq (r(1-\epsilon))^2 + (d-1)(2\epsilon\delta)^2, \tag{since $\bx[i],\bx[j] \geq 1$} \\
    & = (r(1-\epsilon))^2 + 4\epsilon r^2, \tag{by definition of $\epsilon$ in \autoref{defn:r-epsilon-general-d}}\\
    & = (r(1+\epsilon))^2 < (2r)^2. \tag{since $\epsilon < 1$}
  \end{align*}

\item[(ii)] By~\autoref{table:general-d:construction}, the points $C_{\bolda}^{\boldx}$ and $B_{\bolda}^{-i}$ are located at $\bolda\oplus \epsilon\cdot \boldx$ and $\bolda \ominus \bolde_{i}\cdot r(1-\epsilon) \ominus (\bone^d-\be_i)\cdot 2\epsilon \delta$ respectively. Hence, we have
  \begin{align*}
    \dist\left(C_{\ba}^{\bx} , B_{\bolda}^{-i}\right)^2 & = (r(1-\epsilon)+\epsilon\cdot \bx[i])^{2} + \sum_{j=1 \colon j\neq i}^{d} (\epsilon\cdot \bx[j] +2\epsilon\delta)^2, \\
    & \leq (r(1-\epsilon) + \epsilon\delta)^2 + (d-1)(3\epsilon\delta)^2, \tag{since $\bx[i],\bx[j] \leq \delta$}\\
    & = (r(1-\epsilon) + \epsilon\delta)^2 + 9\epsilon r^2, \tag{by definition of $\epsilon$} \\
    & \leq 2r^2(1-\epsilon)^2 + 2\epsilon^2\delta^2 + 9\epsilon r^2, \tag{since $(\alpha+\beta)^2 \leq 2\alpha^2 + 2\beta^2$}\\
    & \leq 2r^2(1-\epsilon)^2 +  11\epsilon r^2, \tag{since $\epsilon\delta^2 \leq r^2$}\\
    & = 2r^2((1-\epsilon)^2 +  5.5\epsilon) < 2r^2(1+1.75\epsilon)^2 < (2r)^2. \tag{since $\epsilon \leq 1/16$}
\end{align*}
\end{enumerate}
\end{proof}

\begin{claim}
\normalfont
For each $\bolda\in \VV$, the distance of $\ba$ to any point \textcolor{black}{in} $\border[\ba]$ is $r(1+\epsilon)$.
\label{lem:D[a]-are-close-to-a}
\end{claim}
\begin{proof}
  Let $p$ be any point in $\border[\ba]$. Then we have two choices for $p$, \textcolor{black}{namely $p=B_{\ba}^{+i}$ or $p=B_{\ba}^{-i}$. In both cases, we have}
  \begin{align*}
    \dist(p,\ba)^2  = (r(1-\epsilon))^2 + (d-1)(2\epsilon\delta)^2 = r^2(1-\epsilon)^2 + 4\epsilon r^2 = (r(1+\epsilon))^2,
  \end{align*}
  \textcolor{black}{where the second equality is obtained by the definition of $\epsilon$} (\autoref{defn:r-epsilon-general-d}).
\end{proof}

\begin{lemma} %$[\star]$\footnote{Proofs of results labeled with $[\star]$ are deferred to the Appendix due to space restrictions}
\normalfont
For each $\bolda\in \VV$ and each $i\in [d]$,
\begin{itemize}
  \item If $w\in \mathcal{U}$ such that $\dist\left( w, B_{\bolda}^{+i}\right)<2r(1+\epsilon)$ then $w\in \left( \mathcal{D}[\bolda]\ \bigcup\ \mathcal{S}_{\{\bolda,\bolda\oplus\bolde_{i} \}}\right)$.
  \item If $w\in \mathcal{U}$ such that $\dist\left( w, B_{\bolda}^{-i}\right)<2r(1+\epsilon)$ then $w\in \left( \mathcal{D}[\bolda]\ \bigcup\ \mathcal{S}_{\{\bolda,\bolda\ominus\bolde_{i} \}}\right)$.
\end{itemize}
\label{lem:border-doesnt-intersect-other-connectors-general-d}
\end{lemma}
\begin{proof}
The proof of this lemma is quite long, and is hence deferred to~\autoref{app:long-lemma-proof} to maintain the flow of the paper.
\end{proof}

\begin{remark}
\normalfont
\autoref{lem:border-doesnt-intersect-other-connectors-general-d} gives a necessary but not sufficient condition. Also, it might be the case that for some $\bolda\in \VV$ and $i\in [d]$ the vector $\bolda\oplus\bolde_{i}\notin \VV \left(\mbox{resp., } \bolda\ominus \bolde_{i}\notin \VV\right)$ in which case the set $\mathcal{S}_{\{\bolda,\bolda\oplus\bolde_{i} \}} \left(\mbox{resp., } \mathcal{S}_{\{\bolda,\bolda\ominus \bolde_{i} \}}\right)$ is empty.
\end{remark}

\begin{lemma}
\normalfont
Let $\bolda\in \VV$ and $i\in [d]$ be such that $\bolda':=(\bolda\oplus\bolde_i) \in \VV$. For each $\ell\in [\delta]$,
\begin{itemize}%[leftmargin=*,align=left]
  \item[ (1)] If $\boldx\in R_{\bolda}$ and $\ell\leq \bx[i]$, then $\dist\left(C_{\ba}^{\bx} , S_{\{\bolda,\bolda'\}}^{\ell}\right)< 2r$.

  \item[ (2)] If $\boldx\in R_{\bolda}$ and $\ell> \bx[i]$, then $\dist\left(C_{\ba}^{\bx} , S_{\{\bolda,\bolda'\}}^{\ell}\right)\geq 2r(1+\epsilon)$.

  \item[ (3)] If $\boldy\in R_{\ba'}$ and $\ell> \by[i]$, then $\dist\left( C_{\ba'}^{\by} , S_{\{\bolda,\bolda'\}}^{\ell}\right)< 2r$.

  \item[ (4)] If $\boldy\in R_{\ba'}$ and $\ell\leq \by[i]$, then $\dist\left( C_{\ba'}^{\by} , S_{\{\bolda,\bolda'\}}^{\ell}\right)\geq 2r(1+\epsilon)$.
\end{itemize}
\label{lem:how-H-intersects-with-geq-general-d}
\end{lemma}
\begin{proof}
Recall from~\autoref{table:general-d:construction} that the points $C_{\ba}^{\bx}$ and $S_{\{\bolda,\bolda'\}}^{\ell}$ are located at $\ba\oplus \epsilon\cdot \bx$ and $\bolda\oplus\bolde_{i}\cdot((1-\epsilon)2r+\epsilon\ell)$ respectively.
\begin{itemize}
  \item[(1)] If $\ell\leq \bx[i]$, then $\dist\left( C_{\ba}^{\bx} , S_{\{\bolda,\bolda'\}}^{\ell} \right)^2$
    \begin{align*}
      &=(2r(1-\epsilon) + \epsilon(\ell-\bx[i]))^2 + \sum_{j=1 \colon j\neq i}^{d} (\epsilon\cdot \bx[j])^2, \\
      & \leq (2r(1-\epsilon))^2 + (d-1)\epsilon^2\delta^2 = (2r(1-\epsilon))^2 + \epsilon r^2 \tag{since $\ell \leq \bx[i]$ and $\bx[j] \leq \delta$} \\
      & = (2r)^2\left((1-\epsilon)^2 + \frac{\epsilon}{4}\right) < (2r)^2. \tag{since $0< \epsilon < 1$}
    \end{align*}

\item[(2)] If $\ell> \bx[i]$, then $\dist\left( C_{\ba}^{\bx} , S_{\{\bolda,\bolda'\}}^{\ell} \right)^2$
  \begin{align*}
    &=(2r(1-\epsilon) + \epsilon(\ell-\bx[i]))^2 + \sum_{j=1 \colon j\neq i}^{d} (\epsilon\cdot \bx[j])^2, \\
    & \geq (2r(1-\epsilon) + \epsilon)^2 = (2r(1-\epsilon) + 4r\epsilon)^2 = (2r(1+\epsilon))^2. \tag{since $\ell > \bx[i]$ and $4r = 1$}
      \end{align*}
\end{itemize}
We now show the remaining two claims: recall from~\autoref{table:general-d:construction} that the points $C_{\ba'}^{\by}$ and $S_{\{\bolda,\bolda'\}}^{\ell}$ are located at $(\ba'\oplus \epsilon\cdot \by) = \ba\oplus \be_i \oplus \epsilon\cdot \by$ and $\bolda\oplus\bolde_{i}\cdot((1-\epsilon)2r + \epsilon\ell)$ respectively.
\begin{itemize}
  \item[(3)] If $\ell> \by[i]$, then $\dist\left( C_{\ba'}^{\by} , S_{\{\bolda,\bolda'\}}^{\ell} \right)^2$
    \begin{align*}
      & = (1+\epsilon\cdot\by[i] - (1-\epsilon)2r -\epsilon\ell)^2 + \sum_{j=1 \colon j\neq i}^{d} (\epsilon\cdot \by[j])^2, \\
      & \leq (4r+\epsilon\cdot\by[i] - (1-\epsilon)2r -\epsilon\ell)^2 + (d-1)\epsilon^2\delta^2, \tag{since $4r=1$ and $\by[j]\leq \delta$}\\
      & = (2r(1+\epsilon) - \epsilon (\ell-\by[i]))^2 + \epsilon r^2, \tag{since $(d-1)\epsilon\delta^2 = r^2$}\\
      & \leq (2r(1+\epsilon) - \epsilon)^2 + \epsilon r^2, \tag{since $\ell > \by[i]$}\\
      & = (2r(1-\epsilon))^2 + \epsilon r^2, \tag{since $4r =1$}\\
      & = (2r)^2\left((1-\epsilon)^2 + \frac{\epsilon}{4}\right) < (2r)^2. \tag{since $0< \epsilon < 1$}
      \end{align*}

  \item[(4)] If $\ell\leq \by[i]$, then $\dist\left( C_{\ba'}^{\by} , S_{\{\bolda,\bolda'\}}^{\ell} \right)^2$
    \begin{align*}
      & = (1+\epsilon\cdot\by[i] - (1-\epsilon)2r -\epsilon\ell)^2 + \sum_{j=1 \colon j\neq i}^{d} (\epsilon\cdot \by[j])^2, \\
      & \geq  (2r(1+\epsilon) + \epsilon (\by[i]-\ell))^2, \tag{since $4r = 1$}\\
      & \geq (2r(1+\epsilon))^2. \tag{since $\by[i] \geq \ell$}
      \end{align*}
\end{itemize}
\end{proof}

\begin{lemma}
\normalfont
Let $\bolda\in \VV$ and $i\in [d]$ be such that $\bolda':=(\bolda\oplus\bolde_i) \in \VV$. If $\ba''\notin \{\ba, \ba'\}$ then the distance between any point \textcolor{black}{in} $\core[\ba'']$ \textcolor{black}{and} any point in $\mathcal{S}_{\ba,\ba'}$ is at least $2r(1+\epsilon)$.
\label{lem:connector-far-away-from-other-cores}
\end{lemma}
\begin{proof}
Let $\bp$ and $\bq$ be two arbitrary points from $\core[\ba'']$ and $\mathcal{S}_{\ba,\ba'}$, respectively. By \autoref{table:general-d:construction}, $\bp$ is located at $\ba''\oplus \epsilon\cdot \bx$ for some $\bx\in R_{\ba}\subseteq [\delta]^d$ and $\bq$ is located at $\ba\oplus \be_i \cdot ((1-\epsilon)2r + \epsilon\ell)$ for some $\ell\in [\delta]$.

Since $\ba'=\ba\oplus \be_i$ and $\ba''\notin \{\ba,\ba'\}$, we have three cases to consider:
\begin{itemize}
\item \underline{$\ba''[j]= \ba[j]$ for all $j\neq i$ and $\ba''[i]\leq \ba[i]-1$}: In this case, we have $\dist(\bp, \bq)^2$
  \begin{align*}
    & \geq \left(\left(\ba[i]+ (1-\epsilon)2r+\epsilon\ell\right) - \left(\ba''[i]+\epsilon\cdot \bx[i]\right) \right)^2, \tag{only considering the $i$-th coordinate}\\
    & = \left(\ba[i]-\ba''[i] + (1-\epsilon)2r+\epsilon \ell -\epsilon \bx[i] \right)^2,\\
    & \geq \left(1 + (1-\epsilon)2r + \epsilon\cdot 4r -\epsilon\delta \right)^2, \tag{since $\ba[i]-\ba''[i] \geq 1$, $\ell\geq 1 = 4r$ and $\bx[i]\leq \delta$}\\
    & > (2r(1+\epsilon))^2. \tag{since $1-\epsilon\delta \geq 1- \frac{1}{16}>0$}
  \end{align*}

\item \underline{$\ba''[j]= \ba[j]$ for all $j\neq i$ and $\ba''[i]\geq \ba[i]+2$}: In this case, we have $\dist(\bp, \bq)^2$
  \begin{align*}
    &\geq \left(\left(\ba''[i] +\epsilon\cdot \bx[i] \right) - \left(\ba[i]+(1-\epsilon)2r+\epsilon\ell \right) \right)^2, \tag{only considering the $i$-th coordinate}\\
    & = \left(\ba''[i] - \ba[i] - (1-\epsilon)2r +\epsilon\cdot \bx[i] - \epsilon\ell \right)^2, \\
    & \geq (2 - (1-\epsilon)2r +\epsilon -\epsilon\delta)^2, \tag{since $\ba''[i]- \ba[i] \geq 2$, $\bx[i] \geq 1$ and $\ell \leq \delta$}\\
    & = (4r - (1-\epsilon)2r + 1 + \epsilon - \epsilon\delta)^2, \tag{since $4r = 1$}\\
    & > (2r(1+\epsilon))^2. \tag{since $1-\epsilon\delta \geq 1- \frac{1}{16}>0$}
  \end{align*}

\item \underline{There exists $j\neq i$ such that $\ba''[j]\neq \ba[j]$}: In this case, we have $\dist(\bp, \bq)$
  \begin{align*}
    &\geq \left|\ba[j]-\left(\ba''[j]+\epsilon\cdot \bx[j]\right)\right|, \tag{only considering the $j$-th coordinate}\\
    &\geq \left|\ba[j]-\ba''[j]\right| - \epsilon\cdot \bx[j], \tag{by triangle inequality}\\
    &\geq 1-\epsilon\cdot \delta, \tag{since $\ba[j]\neq \ba''[j]$ and $\bx[j]\leq\delta$} \\
    & \geq 2r + 2r - r^2 = 2r + 2r \left(1-\frac{r}{2}\right), \tag{since $4r = 1$ and $\epsilon\delta \leq r^2$}\\
    & > 2r(1+\epsilon). \tag{since $1-\frac{r}{2}>\frac{1}{16}\geq \epsilon$}
  \end{align*}
\end{itemize}
\end{proof}

\subsection{{\normalsize
$\mathcal{I}$ has a satisfying assignment $\Rightarrow$ \OPT for the instance $\mathcal{U}$ of $|\VV|$-\kcenter is $< 2r$
}}
\label{subsec:k-center-general-d-easy}

Suppose that the $d$-dimensional geometric \geqcsp $\mathcal{I}=(\VV,\DD,\CC)$ has a satisfying assignment $f:\VV\to \DD$. Consider the set of points $F$ given by $\Big \{ C_{\bolda}^{f(\bolda)} : \bolda\in \VV\Big\}$. Since $f:\VV\to \DD$ is a satisfying assignment for $\mathcal{I}$, it follows that $f(\bolda)\in R_{\bolda}$ for each $\bolda \in \VV$ and hence the set $F$ is well-defined. Clearly, $|F|=|\VV|$. We now show that
$$\OPT(F):=\Big(\max_{u\in \mathcal{U}} \big( \min_{v\in F} \dist(u,v) \big) \Big)< 2r$$
This implies that \OPT for the instance $\mathcal{U}$ of $|\VV|$-\kcenter is $< 2r$. We show $\OPT(F)<2r$ by showing that $\dist(p, F)<2r$ for each $p\in \mathcal{U}$. From~\autoref{table:general-d:construction} and~\autoref{table-special-sets-general-d}, it is sufficient to consider the two cases depending on whether $p$ is a primary point or a secondary point.

\begin{lemma}
\normalfont
If $p$ is a primary point, then $\dist(p, F)<2r$.
\label{lem:general-d-udg-dominates-D}
\end{lemma}
\begin{proof}
If $p$ is a primary point, then by~\autoref{table:general-d:construction} and~\autoref{table-special-sets-general-d} it follows that $p$ is either a core point or a border point:
\begin{itemize}
  \item \textbf{$p$ is a core point}: By~\autoref{table:general-d:construction}, $p\in \core[\bb]$ for some $\bb\in \VV$. Then, \autoref{lem:core-pairwise-intersects-general-d} implies that $\dist\big(p, C_{\bb}^{f(\bb)} \big) <r$. Since $C_{\bb}^{f(\bb)}\in F$, we have $\dist\big( p, F \big) \leq \dist\big(p, C_{\bb}^{f(\bb)} \big) <r$.
  \item \textbf{$p$ is a border point}: By~\autoref{table:general-d:construction}, $p\in \border[\bb]$ for some $\bb\in \VV$. Then, \autoref{lem:core-intersects-all-border-general-d} implies that $\dist\big( p, C_{\bb}^{f(\bb)}\big) <2r$. Since $C_{\bb}^{f(\bb)}\in F$, we have $\dist\big( p, F \big)\leq \dist\big(p, C_{\bb}^{f(\bb)} \big) <2r$. \qedhere
\end{itemize}
\end{proof}

\begin{lemma}
\normalfont
If $p$ is a secondary point, then $\dist(p, F)<2r$.
\label{lem:secondary-at-most-2r}
\end{lemma}
\begin{proof}
If $p$ is a secondary point, then by~\autoref{table:general-d:construction} and~\autoref{table-special-sets-general-d} it follows that there exists $\ba\in \VV, i\in [d]$ and $\ell\in [\delta]$ such that $p = S_{\{\ba, \ba\oplus \be_i\}}^{\ell}$. Note that $C_{\ba}^{f(\ba)}\in F$ and $C_{\ba\oplus \be_i}^{f(\ba\oplus \be_i)}\in F$. We now prove the lemma by showing that $\min \Big\{ \dist \big(p, C_{\ba}^{f(\ba)} \big) ; \dist \big(p, C_{\ba\oplus \be_i}^{f(\ba\oplus \be_i)} \big) \Big\}<2r$. Since $f:\VV\to \DD$ is a satisfying assignment, the binary constraint on $\ba$ and $\ba\oplus \be_i$ is satisfied, i.e., $\delta\geq f(\ba)[i]\geq f(\ba\oplus \be_i)[i]\geq 1$. Since $\ell\in [\delta]$ this implies that either $\ell\leq f(\ba)[i]$ or $\ell> f(\ba\oplus \be_i)[i]$. The following implications complete the proof:
\begin{itemize}
  \item If $\ell\leq f(\ba)[i]$, then~\autoref{lem:how-H-intersects-with-geq-general-d}(1) implies that $\dist\big( C_{\ba}^{f(\ba)}, p\big) < 2r$.

  \item If $\ell> f(\ba\oplus \be_i)[i]$, then~\autoref{lem:how-H-intersects-with-geq-general-d}(3) implies that $\dist\big( C_{\ba\oplus \be_i}^{f(\ba\oplus \be_i)}, p\big) < 2r$. \qedhere
\end{itemize}
\end{proof}

\noindent  From~\autoref{table-special-sets-general-d},~\autoref{lem:general-d-udg-dominates-D} and~\autoref{lem:secondary-at-most-2r} it follows that \OPT for the instance $\mathcal{U}$ of $|\VV|$-\kcenter is $<2r$.

\subsection{{\normalsize $\mathcal{I}$ does not have a satisfying assignment $\Rightarrow$ \OPT for the instance $\mathcal{U}$ of $|\VV|$-\kcenter is $\geq 2r(1+\epsilon)$}}
%\subsection{{\normalsize $\mathcal{I}$ has a dominating set of size $\leq k^2$ $\Rightarrow$ \gtgeq has a solution}}
\label{subsec:k-center-general-d-hard}

Suppose that the instance $\mathcal{I}=(\VV,\DD,\CC)$ of $d$-dimensional geometric \geqcsp does not have a satisfying assignment. We want to now show that \OPT for the instance $\mathcal{U}$ of $|\VV|$-\kcenter is $\geq 2r(1+\epsilon)$. Fix any set $Q\subseteq \mathcal{U}$ of size $|\VV|$: it is sufficient to show that
\begin{equation}\label{eqn:opt-less-2r-eps}
\OPT(Q):=\Big(\max_{u\in \mathcal{U}} \big( \min_{v\in Q} \dist(u,v) \big) \Big)\geq 2r(1+\epsilon)
\end{equation}
We consider two cases: either $\big|Q\cap \core[\ba]\big|=1$ for each $\ba\in \VV$ (\autoref{lem:exactly-one-per-core}) or not (\autoref{lem:not-exactly-one-per-core}).

\begin{lemma}
\normalfont
If $\big|Q\cap \core[\ba]\big|=1$ for each $\ba\in \VV$ then $\OPT(Q)\geq 2r(1+\epsilon)$.
\label{lem:exactly-one-per-core}
\end{lemma}
\begin{proof}
Since $|Q|=|\VV|$ and $|Q\cap \core[\ba]|=1$ for each $\ba\in \VV$ it follows that the only points in $Q$ are core points (see~\autoref{table:general-d:construction} for definition) and moreover $Q$ contains exactly one core point corresponding to each element from $\VV$. Let $\phi: \VV \to [\delta]^d$ be the function such that $Q\cap \core[\ba] = C_{\ba}^{\phi(\ba)}$. By~\autoref{table:general-d:construction}, it follows that $\phi(a)\in \R_{\ba}$ for each $\ba\in \VV$.

Recall that we are assuming in this section that the instance $\mathcal{I}=(\VV,\DD,\CC)$ of $d$-dimensional geometric \geqcsp does not have a satisfying assignment. Hence, in particular, the function $\phi: \VV \to [\delta]^d$ is not a satisfying assignment for $\mathcal{I}$. All unary constraints are satisfied since $\phi(a)\in \R_{\ba}$ for each $\ba\in \VV$. Hence, there is some binary constraint which is not satisfied by $\phi$: let this constraint be violated for the pair $\ba, \ba\oplus \be_i$ for some $\ba\in \VV$ and $i\in [d]$. Let us denote $\ba\oplus \be_i$ by $\ba'$. The violation of the binary constraint on $\ba$ and $\ba\oplus \be_i$ by $\phi$ implies that $1\leq \phi(\ba)[i]<\phi(\ba')[i]\leq \delta$. We now show that $\dist\big(Q, S_{\{\ba,\ba'\}}^{\phi(\ba')[i]} \big)\geq 2r(1+\epsilon)$ which, in turn, implies that $\OPT(Q)\geq 2r(1+\epsilon)$. The following implications complete the proof:
\begin{itemize}
  \item \autoref{lem:how-H-intersects-with-geq-general-d}(2) implies that $\dist\big( S_{\{\ba,\ba'\}}^{\phi(\ba')[i]}, C_{\ba}^{\phi(\ba)} \big) \geq 2r(1+\epsilon)$.

  \item \autoref{lem:how-H-intersects-with-geq-general-d}(4) implies that $\dist\big( S_{\{\ba,\ba'\}}^{\phi(\ba')[i]}, C_{\ba'}^{\phi(\ba')} \big) \geq 2r(1+\epsilon)$.

  \item Consider any point $s\in Q\setminus \big\{ C_{\ba}^{\phi(\ba)},  C_{\ba'}^{\phi(\ba')} \big\}$. Then $s\in \core[\ba'']$ for some $\ba''\notin \big\{ \ba, \ba' \big\}$. \autoref{lem:connector-far-away-from-other-cores} implies $\dist\big( S_{\{\ba, \ba'\}}^{\phi(\ba')[i]}, s \big)\geq 2r(1+\epsilon)$. \qedhere
\end{itemize}
\end{proof}

\begin{lemma}
\normalfont
If there exists $\ba\in \VV$ such that $\big|Q\cap \core[\ba]\big|\neq 1$ then $\OPT(Q)\geq 2r(1+\epsilon)$.
\label{lem:not-exactly-one-per-core}
\end{lemma}
\begin{proof}
Suppose that $\OPT(Q)< 2r(1+\epsilon)$. To prove the lemma, we will now show that this implies $|Q\cap \core[\ba]|=1$ for each $\ba\in \VV$. This is done via the following two claims, namely~\autoref{clm:exactly-one-D} and~\autoref{clm:exactly-one-core}.

\begin{claim}
$\big|Q\cap \mathcal{D}[\ba]\big|=1$ for each $\ba\in \VV$
\label{clm:exactly-one-D}
\end{claim}
\begin{proof}
Define three sets $I_0, I_1$ and $I_{\geq 2}$ as follows:
\begin{eqnarray}
  I_0 := \big\{ \ba\in \VV : \big|Q\cap \mathcal{D}[\ba]\big|=0 \big\}  \\
  I_1 := \big\{ \ba\in \VV : \big|Q\cap \mathcal{D}[\ba]\big|=1 \big\} \\
  I_{\geq 2} := \big\{ \ba\in \VV : \big|Q\cap \mathcal{D}[\ba]\big|\geq 2 \big\}
\end{eqnarray}
By definition, we have
\begin{equation}
|I_0|+ |I_1| + |I_{\geq 2}| = |\VV|
\label{eqn:sum-of-indices-is-k^2}
\end{equation}

Consider a variable $\boldb\in I_0$. \textcolor{black}{Since $\dist\left(Q,B_{\boldb}^{+i}\right)$ and $\dist\left(Q,B_{\boldb}^{-i}\right)<2r(1+\epsilon)$, and $Q\cap \mathcal{D}[\boldb]=\emptyset$}, \autoref{lem:border-doesnt-intersect-other-connectors-general-d} implies that for each $i\in [d]$
\begin{enumerate}%[leftmargin=*,align=left]
\item[(i)] $Q$ must contain a point from $\mathcal{S}_{\{\boldb,\boldb\oplus\bolde_i\}}$ since $Q\cap \mathcal{D}[\boldb]=\emptyset$, and

\item[(ii)] $Q$ must contain a point from $\mathcal{S}_{\{\boldb,\boldb\ominus\bolde_i\}}$ since $Q\cap \mathcal{D}[\boldb]=\emptyset$
\end{enumerate}
Since each secondary point can be \emph{``charged''} to two variables in $\VV$ (recall the definition of secondary points from~\autoref{table:general-d:construction}: each secondary point is indexed by a set of two variables $\{\bb,\bb'\}$ such that $\bb'=\bb\oplus \be_i$ for some $i\in [d]$), it follows that $Q$ contains $\geq \frac{2d}{2}=d\geq 2$ \emph{distinct} secondary points corresponding to each variable in $I_{0}$.
Therefore, we have
\begin{align}
&|I_0|+|I_1|+|I_{\geq 2}| = \big|\VV\big| \tag{from~\autoref{eqn:sum-of-indices-is-k^2}}\\
&= |Q| \tag*{} \\
&\geq |Q\cap \textsc{Primary}| + |Q\cap \textsc{Secondary}| \tag{since $\textsc{Primary}\cap \textsc{Secondary}=\emptyset$ }\\
&\geq \Big(|I_1| + 2|I_{\geq 2}|\Big) + |Q\cap \textsc{Secondary}| \tag{by definition of $I_1$ and $I_{\geq 2}$}\\
&\geq \Big(|I_1| + 2|I_{\geq 2}|\Big) + 2|I_0|
\end{align}
\textcolor{black}{where the last inequality follows because $Q$ contains at least $2$ secondary points corresponding to each variable in $I_{0}$}.
Hence, we have $|I_0|+|I_1|+|I_{\geq 2}| \geq 2|I_0| + |I_1| + 2|I_{\geq 2}|$ which implies $|I_0|=0=|I_{\geq 2}|$. From~\autoref{eqn:sum-of-indices-is-k^2}, we get $|I_1|=\big|\VV\big|$, i.e., $\big|Q\cap \mathcal{D}[\ba]\big|=1$ for each $\ba\in \VV$. This concludes the proof of~\autoref{clm:exactly-one-D}.
\end{proof}

Since $|Q|=\big|\VV\big|$ \textcolor{black}{and $\mathcal{D}[\ba] \cap \mathcal{D}[\bb] = \emptyset$ for distinct $\ba,\bb \in \VV$},~\autoref{clm:exactly-one-D} implies that
\begin{equation}\label{eqn:Q-has-no-secondary-pts}
Q\ \text{contains no secondary points}
\end{equation}

\textcolor{black}{We now prove that $Q$ doesn't contain border points either}.

\begin{claim}
$\big|Q\cap \core[\ba]\big|=1$ for each $\ba\in \VV$
\label{clm:exactly-one-core}
\end{claim}
\begin{proof}
Fix any $\ba\in \VV$. From~\autoref{clm:exactly-one-D}, we know that $\big|Q\cap \mathcal{D}[\ba]\big|=1$. Suppose that this unique point in $Q\cap \mathcal{D}[\ba]$ is from $\border[\ba]$. Without loss of generality, let $Q\cap \mathcal{D}[\ba]=\big\{B_{\ba}^{+i}\big\}$ for some $i\in [d]$. Since $\OPT(Q)<2r(1+\epsilon)$, it follows that $\dist\big(Q, B_{\ba}^{-i}\big)<2r(1+\epsilon)$. Hence,~\autoref{lem:border-doesnt-intersect-other-connectors-general-d}(2) implies that $Q\cap \Big( \mathcal{D}[\bolda]\ \bigcup\ \mathcal{S}_{\{\bolda,\bolda\ominus\bolde_{i} \}}\Big) \neq \emptyset$.
\textcolor{black}{Since $Q$ contains no secondary points (\autoref{eqn:Q-has-no-secondary-pts}), we have $ Q\cap \left( \mathcal{D}[\bolda]\ \bigcup\ \mathcal{S}_{\{\bolda,\bolda\ominus\bolde_{i} \}}\right) = Q\cap \mathcal{D}[\ba] = \left\{B_{\ba}^{+i}\right\}$. But from \autoref{lem:borders-dont-intersect-2d} we know $\dist \left( B_{\ba}^{+i}, B_{\ba}^{-i}\right)\geq 2r(1+\epsilon)$. We thus obtain a contradiction.}
This concludes the proof of~\autoref{clm:exactly-one-core}.
\end{proof}

Therefore, we have shown that $\OPT(Q)< 2r(1+\epsilon)$ implies $\big|Q\cap \core[\ba]\big|=1$ for each $\ba\in \VV$. This concludes the proof of~\autoref{lem:not-exactly-one-per-core}.
\end{proof}

\subsection{Finishing the proof of~\autoref{thm:dom-set-d-dimensions}}
\label{subsec:finishing-the-proof}

Finally, we are ready to prove~\autoref{thm:dom-set-d-dimensions} which is restated below:

\domsetd*

\begin{proof}

Given an instance $\mathcal{I}=(\VV,\DD,\CC)$ of a $d$-dimensional geometric \geqcsp, we build an instance $\mathcal{U}$ of $|\VV|$-\kcenter in $\mathbb{R}^d$ given by the reduction in~\autoref{subsec:redn-general-d}. This reduction has the property that
\begin{itemize}
  \item if $\mathcal{I}$ has a satisfying assignment then \OPT for the instance $\mathcal{U}$ of $|\VV|$-\kcenter is $< 2r$ (\autoref{subsec:k-center-general-d-easy}), and
  \item if $\mathcal{I}$ does not have a satisfying assignment then \OPT for the instance $\mathcal{U}$ of $|\VV|$-\kcenter is $\geq 2r(1+\epsilon^*)$ (\autoref{subsec:k-center-general-d-hard})
\end{itemize}
where $r=1/4$ and \textcolor{black}{$\epsilon^* = \dfrac{r^2}{(d-1)\delta^2}\geq \dfrac{1}{16(d-1)|\DD|}$, since $|\DD|=\left|[\delta]^d\right| \geq \delta^2$}. Hence, any algorithm for the $|\VV|$-center problem which has an approximation factor $\leq (1+\epsilon^*)$ can solve the $d$-dimensional geometric \geqcsp.
Note that the instance $\mathcal{U}$ of $k$-\kcenter in $\mathbb{R}^d$ has $k=|\VV|$ and the number of points $n\leq |\VV|\cdot 2d + |\CC| +|\VV|^{2}\cdot \delta = |\mathcal{I}|^{O(1)}$ where $|\mathcal{I}|=|\VV|+|\DD|+|\CC|$. We now derive the two lower bounds claimed in the theorem:
\begin{description}
  \item[\textbf{- (Inapproximability result)}]  Suppose that there exists $d\geq 2$ such that the $k$-center on $n$ points in $\mathbb{R}^d$ admits an $(1+\epsilon)$-approximation algorithm in $f(k)\cdot \Big(\frac{1}{\epsilon}\Big)^{o(k^{1-1/d})}\cdot n^{o(k^{1-1/d})}$ time for some computable function $f$. As argued above, using a $(1+\epsilon^*)$-approximation for the $k$-center problem with $k=|\VV|$ and $n=|\mathcal{I}|^{O(1)}$ points can solve the $d$-dimensional geometric \geqcsp problem. Recall that $16(d-1)|\cdot  \mathcal{I}|\geq 16(d-1)|\cdot \DD|\geq \frac{1}{\epsilon^*}$ since $|I|=|\VV|+|\DD|+|\CC|$, and hence we have an algorithm for the $d$-dimensional geometric \geqcsp problem which runs in time $f(|\VV|)\cdot (16d)^{o(k^{1-1/d})}\cdot |\mathcal{I}|^{o(k^{1-1/d})}$ which contradicts~\autoref{thm:marx-sidiropoulos}.

  \item[\textbf{- (Lower bound for exact algorithm)}] Suppose that there exists $d\geq 2$ such that the $k$-center on $n$ points in $\mathbb{R}^d$ admits an exact algorithm in $f(k)\cdot n^{o(k^{1-1/d})}$ time for some computable function $f$. As argued above\footnote{The argument above is actually stronger: even a $(1+\epsilon^*)$-approximation algorithm for $k$-center can solve $d$-dimensional geometric \geqcsp}, solving the $k$ center problem with $k=|\VV|$ and $n=|\mathcal{I}|^{O(1)}$ points can solve the $d$-dimensional geometric \geqcsp problem. Hence, we have an algorithm for the $d$-dimensional geometric \geqcsp problem which runs in time $f(|\VV|)\cdot |\mathcal{I}|^{o(k^{1-1/d})}$ which again contradicts~\autoref{thm:marx-sidiropoulos}. \qedhere
\end{description}
\end{proof}

\newpage

\bibliography{references}
\bibliographystyle{plainnat}
\newpage

\appendix

\section{Proof of~\autoref{lem:border-doesnt-intersect-other-connectors-general-d}}
\label{app:long-lemma-proof}

\begin{proof}
  Fix some $\bolda\in \VV$ and $i\in [d]$. To show the first part of lemma, it suffices to show that if $w\in \mathcal{U}\setminus \left( \mathcal{D}[\bolda]\ \bigcup\ \mathcal{S}_{\{\bolda,\bolda\oplus\bolde_{i} \}}\right)$ then $\dist\left( w, B_{\bolda}^{+i}\right) \geq 2r(1+\epsilon)$. We argue based on whether $w$ is a primary or secondary point. Recall, $B_{\ba}^{+i} = \ba\oplus\bolde_i\cdot r(1-\epsilon)\oplus (\bone^d-\be_i)\cdot 2\epsilon \delta$.

We begin with the following two cases which arise when $w$ is a \textbf{primary} point.
\begin{enumerate}
\item[(1)] \underline{$w\in \core[\bolda']$ for some $\bolda'\neq \bolda$}: From \autoref{table:general-d:construction}, we know $w=C_{\ba'}^{\bx}$ for some $\bx\in R_{\ba'}\subseteq [\delta]^d$. That is, $w = \ba'\oplus\epsilon\cdot\bx$. Since $\ba\neq \ba'$, it follows that $\ba$ and $\ba'$ differ in at least one coordinate, say the $j$-th coordinate.
  \begin{itemize}
  \item[(i)] \underline{$i\neq j$}: In this case, we have
    \begin{align*}
      \dist\left(w, B_{\ba}^{+i} \right)
      &\geq \big|(\ba'[j]+\epsilon\cdot \bx[j]) -(\ba[j] +2\epsilon\delta) \big|, \tag{only considering the $j$-th coordinate}\\
      &\geq \big|\ba'[j]-\ba[j]\big| - \big| 2\epsilon\delta -\epsilon\cdot \bx[j] \big|, \tag{by triangle inequality}\\
      &\geq 1- 2\epsilon\delta, \tag{since $\ba'[j] \neq \ba[j]$ and $\bx[j] \in [\delta$]}\\
      &\geq 2r(1+\epsilon). \tag{from~\autoref{defn:r-epsilon-general-d} and~\autoref{eqn:always-to-be-cited}}
    \end{align*}

  \item[(ii)] \underline{$i=j$}: In this case, we have
    \begin{align*}
      \dist\left(w, B_{\ba}^{+i} \right)
      &\geq \big|\ba'[i]+\epsilon\cdot \bx[i] -\ba[i] - r(1-\epsilon) \big|, \tag{only considering the $i$-th coordinate}\\
      &\geq  \big|\ba'[i]-\ba[i]\big| - \big| r(1-\epsilon) -\epsilon\cdot \bx[i] \big|, \tag{by triangle inequality}\\
      &\geq 4r- r(1-5\epsilon), \tag{since $\ba'[i]\neq \ba[i]$, $4r =1$, and $\bx[i]\geq 1$}\\
      &\geq 2r(1+\epsilon). \tag{from~\autoref{defn:r-epsilon-general-d} and~\autoref{eqn:always-to-be-cited}}
    \end{align*}
  \end{itemize}

\item[(2)] \underline{$w\in \border[\bolda']$ for some $\bolda'\neq \bolda$}: From~\autoref{table:general-d:construction}, $w\in \left\{ B_{\ba'}^{+q}, B_{\ba'}^{-q} \right\}$ for some $q\in [d]$. Since $\ba\neq \ba'$, they differ in at least one coordinate, say the $j$-th coordinate. We now consider two cases depending on $\dist\left(\ba ,\ba'\right)$.
  \begin{itemize}
  \item[(i)] \underline{$\dist\left(\ba, \ba'\right)\geq \sqrt{2}$}:
    By the triangle inequality we have
    \[\sqrt{2}\leq \dist\left(\ba,\ba'\right)\leq \dist\left(\ba, B_{\ba}^{+i}\right) + \dist\left(B_{\ba}^{+i}, w\right) + \dist\left(w, \ba'\right).\]
    By \autoref{lem:D[a]-are-close-to-a}, it follows that $\dist\left(\ba, B_{\ba}^{+i}\right) = r(1+\epsilon)$ and $\dist(w, \ba') = r(1+\epsilon)$ (since $w\in \border[\ba']$).
    Hence, we have $\dist(B_{\ba}^{+i}, w) \geq  \sqrt{2}-2r(1+\epsilon) = 4\sqrt{2}r - 2r(1+\epsilon) \geq 2r(1+\epsilon)$, where in the last bound we have used \autoref{defn:r-epsilon-general-d} and \autoref{eqn:always-to-be-cited}.

  \item[(ii)] \underline{$\dist\left(\ba, \ba'\right)=1$}: In this case, we have $\dist\left(\ba, \ba'\right)= \big|\ba[j]-\ba'[j]\big|=1$.
    Recall $w\in \left\{ B_{\ba'}^{+q}, B_{\ba'}^{-q} \right\}$, where
    \begin{align*}
    B_{\ba'}^{+q} & = \bolda' \oplus \bolde_{q}\cdot r(1-\epsilon) \oplus (\bone^d-\be_q)\cdot 2\epsilon \delta \quad \text{ and } \\
    B_{\ba'}^{-q} & = \bolda' \ominus \bolde_{q}\cdot r(1-\epsilon) \ominus (\bone^d-\be_q)\cdot 2\epsilon \delta.
    \end{align*}
    Further recall $B_{\ba}^{+i} = \ba\oplus\bolde_i\cdot r(1-\epsilon)\oplus (\bone^d-\be_i)\cdot 2\epsilon \delta$. There are four cases to consider:
    \begin{itemize}
    \item \underline{$j\notin \{i,q\}$}: In this case, we have
      \begin{align*}
        \dist\left(w, B_{\ba}^{+i}\right) & \geq \big|\ba'[j]-\ba[j] -4\epsilon\delta \big|, \tag{only considering the $j$-th coordinate}\\
        &\geq 1 -4\epsilon\delta, \tag{by triangle inequality and $\ba[j]\neq \ba'[j]$}\\
        &\geq 2r(1+\epsilon). \tag{from~\autoref{defn:r-epsilon-general-d} and~\autoref{eqn:always-to-be-cited}}
      \end{align*}

    \item \underline{$j=i=q$}: In this case, we have
      \begin{align*}
        \dist(w, B_{\ba}^{+i}) & \geq \big|\ba'[j]-\ba[j] -2r(1-\epsilon) \big|, \tag{only considering the $j$-th coordinate}\\
        &\geq 4r -2r(1-\epsilon), \tag{by triangle inequality, $\ba[j]\neq\ba'[j]$, and $4r=1$}\\
        &\geq 2r(1+\epsilon).
      \end{align*}

    \item \underline{$j=i$ and $i\neq q$}: In this case, we have
      \begin{align*}
        \dist(w, B_{\ba}^{+i})& \geq \big|\ba'[j]-\ba[j] -r(1-\epsilon) -2\epsilon\delta \big|, \tag{only considering the $j$-th coordinate}\\
        &\geq 1 -r(1-\epsilon)-2\epsilon\delta, \tag{since $\ba[j]\neq \ba'[j]$}\\
        &\geq 2r(1+\epsilon). \tag{from~\autoref{defn:r-epsilon-general-d} and~\autoref{eqn:always-to-be-cited}}
      \end{align*}

    \item \underline{$j=q$ and $q\neq i$}: In this case, we have
      \begin{align*}
        \dist(w, B_{\ba}^{+i}) & \geq \big|\ba'[j]-\ba[j] -r(1-\epsilon) -2\epsilon\delta \big|, \tag{only considering the $j$-th coordinate}\\
        &\geq 2r(1+\epsilon). \tag{same as the previous case}
      \end{align*}
    \end{itemize}
  \end{itemize}
\end{enumerate}

We now consider the following three cases which arise when $w$ is a \textbf{secondary} point.
\begin{enumerate}
\item[(3)] \underline{$w\in \mathcal{S}_{\{\bolda,\bolda^*\}}$ for some $\bolda^*\neq (\bolda\oplus\bolde_i)$}: Let $w=S_{\{\bolda,\bolda^*\}}^{\ell}$ for some $\ell\in [\delta]$. By definition of \textsc{Secondary} (see~\autoref{table-special-sets-general-d}), the variables $\bolda$ and $\bolda^*$ must form an edge in the constraint graph $G_{\mathcal{I}}$. Since $G_{\mathcal{I}}$ is a subgraph of the $d$-dimensional grid, it suffices to consider the following three cases:
  \begin{enumerate}
  \item[(i)] \underline{$\bolda^* = \bolda\ominus \bolde_i$} : From~\autoref{table:general-d:construction}, the point $w=S_{\{\bolda^*,\bolda\}}^{\ell}$ is located at $\bolda^* \oplus \bolde_{i}\cdot \left((1-\epsilon)2r+\epsilon\ell\right) = (\bolda\ominus \bolde_i) \oplus \bolde_{i}\cdot \left((1+\epsilon)2r+\epsilon\ell\right)$ and the point $B_{\bolda}^{+i}$ is located at $\bolda \oplus \bolde_{i}\cdot r(1-\epsilon) \oplus (\bone^d-\be_i)\cdot 2\epsilon \delta$. Therefore, the distance between $w$ and $B_{\bolda}^{+i}$ is
    \begin{align*}
      &\geq \left|\left(r(1-\epsilon) + 2r -\epsilon\cdot(\ell-2r)\right)\right|, \tag{considering only the $i$-th coordinate and using $4r=1$}\\
      & = \left| (2r(1+\epsilon) +r(1-\epsilon) -\epsilon\ell) \right|, \\
      &> 2r(1+\epsilon). \tag{from $\ell \leq \delta$, \autoref{defn:r-epsilon-general-d} and \autoref{eqn:always-to-be-cited}}
    \end{align*}

  \item[(ii)] \underline{$\bolda^* = \bolda\oplus\bolde_j$ for some $j\neq i$} : From~\autoref{table:general-d:construction}, the point $w=S_{\{\bolda^*,\bolda\}}^{\ell}$ is located at $\bolda \oplus \bolde_{j}\cdot \big(2r+\epsilon\cdot(\ell-2r)\big)$ and the point $B_{\bolda}^{+i}$ is located at $\bolda \oplus \bolde_{i}\cdot (r-\epsilon/4) \oplus (\bone^d-\be_i)\cdot 2\epsilon \delta$. Therefore, square of the distance between $w$ and $B_{\bolda}^{+i}$ is
            \begin{align*}
                &\geq \big(2r+\epsilon\cdot(\ell-2r)-2\epsilon\delta\big)^{2} +(r-\epsilon/4)^2 \tag{ignoring all coordinates except $i$-th \& $j$-th}\\
                &\geq \big(2r-2\epsilon\delta\big)^{2} +(r-\epsilon/4)^2 \tag{from $\ell\geq 1$ and~\autoref{defn:r-epsilon-general-d}} \\
                &> \big(2r(1+\epsilon)\big)^2 \tag{from~\autoref{defn:r-epsilon-general-d} and~\autoref{eqn:always-to-be-cited}}
                \end{align*}

        \item[(iii)] \underline{$\bolda^* = \bolda\ominus \bolde_j$ for some $j\neq i$} : From~\autoref{table:general-d:construction}, the point $w=S_{\{\bolda^*,\bolda\}}^{\ell}$ is located at $\bolda^* \oplus \bolde_{j}\cdot \big(2r+\epsilon\cdot(\ell-2r)\big) = (\bolda\ominus \bolde_j) \oplus \bolde_{j}\cdot \big(2r+\epsilon\cdot(\ell-2r)\big)$ and the point $B_{\bolda}^{+i}$ is located at $\bolda \oplus \bolde_{i}\cdot (r-\epsilon/4) \oplus (\bone^d-\be_i)\cdot 2\epsilon \delta$. Therefore, square of the distance between $w$ and $B_{\bolda}^{+i}$ is
            \begin{align*}
                &\geq \big(2\epsilon\delta +2r-\epsilon\cdot(\ell-2r)\big)^{2} + (r-\epsilon/4)^2 \tag{only considering coordinates $i$ and $j$}\\
                &> ( 2r+\epsilon\cdot\delta)^{2}+ (r-\epsilon/4)^2 \tag{from $\ell\leq \delta$,~\autoref{defn:r-epsilon-general-d} and~\autoref{eqn:always-to-be-cited}}\\
                %&= \Big(2r-\frac{1}{8d\delta}\Big)^{2}+ r^2 \tag{since $\epsilon=\frac{1}{8d\delta^2}$}\\
                &> \big(2r(1+\epsilon)\big)^{2} \tag{from~\autoref{defn:r-epsilon-general-d} and~\autoref{eqn:always-to-be-cited}}
                \end{align*}
        \end{enumerate}

\item[(4)] \underline{$w\in \mathcal{S}_{\{\bolda\oplus\bolde_i,\bolda^{**}\}}$ for some $\bolda^{**}\neq \bolda$}: Let $w=S_{\{\bolda\oplus\bolde_i,\bolda^{**}\}}^{\ell}$ for some $\ell\in [\delta]$. By definition of \textsc{Secondary} (see~\autoref{table-special-sets-general-d}), the variables $\bolda\oplus\bolde_i$ and $\bolda^{**}$ must form an edge in the constraint graph $G_{\mathcal{I}}$. Since $G_{\mathcal{I}}$ is a subgraph of the $d$-dimensional grid, it suffices to consider the following three cases:
        \begin{enumerate}
        \item[(i)] \underline{$\bolda^{**} = (\bolda\oplus\bolde_i)\oplus\bolde_i = \bolda\oplus 2\bolde_i$} : From~\autoref{table:general-d:construction}, the point $w=S_{\{\bolda^{**},\bolda\oplus\bolde_i\}}^{\ell}$ is located at $(\bolda\oplus\bolde_i) \oplus \bolde_{i}\cdot \big(2r+\epsilon\cdot(\ell-2r)\big)$ and the point $B_{\bolda}^{+i}$ is located at $\bolda \oplus \bolde_{i}\cdot (r-\epsilon/4) \oplus (\bone^d-\be_i)\cdot 2\epsilon \delta$. Therefore, square of the distance between $w$ and $B_{\bolda}^{+i}$ is
                \begin{align*}
                &\geq \big(1+2r+\epsilon\cdot(\ell-2r)-(r-\epsilon/4)\big)^{2} \tag{ignoring all coordinates except $i$-th}\\
                &> \big(1+2r-(r-\epsilon/4)\big)^{2} \tag{since $\ell\geq 1> 2r$}\\
                &> \big(2r(1+\epsilon)\big)^{2} \tag{from~\autoref{defn:r-epsilon-general-d} and~\autoref{eqn:always-to-be-cited}}
                \end{align*}
        \item[(ii)] \underline{$\bolda^{**} = (\bolda\oplus \bolde_i)+\bolde_j$ for some $j\neq i$} : From~\autoref{table:general-d:construction}, the point $w=S_{\{\bolda^{**},\bolda\oplus\bolde_i\}}^{\ell}$ is located at $(\bolda\oplus\bolde_i) \oplus \bolde_{j}\cdot \big(2r+\epsilon\cdot(\ell-2r)\big)$ and the point $B_{\bolda}^{+i}$ is located at $\bolda \oplus \bolde_{i}\cdot (r-\epsilon/4) \oplus (\bone^d-\be_i)\cdot 2\epsilon \delta$. Therefore, square of the distance between $w$ and $B_{\bolda}^{+i}$ is
            \begin{align*}
              &\geq \big(1-(r-\epsilon/4)\big)^{2} \tag{ignoring all coordinates except $i$-th}\\
              &> \big(2r(1+\epsilon)\big)^2 \tag{from~\autoref{defn:r-epsilon-general-d} and~\autoref{eqn:always-to-be-cited}}
            \end{align*}

        \item[(iii)] \underline{$\bolda^{**} = (\bolda\oplus\bolde_i)\ominus\bolde_j$ for some $j\neq i$} : From~\autoref{table:general-d:construction}, the point $w=S_{\{\bolda^{**},\bolda\oplus\bolde_i\}}^{\ell}$ is located at $(\bolda\oplus\bolde_i) \ominus \bolde_{j}\cdot \big(2r+\epsilon\cdot(\ell-2r)\big)$ and the point $B_{\bolda}^{+i}$ is located at $\bolda \oplus \bolde_{i}\cdot (r-\epsilon/4) \oplus (\bone^d-\be_i)\cdot 2\epsilon \delta$. Therefore, square of the distance between $w$ and $B_{\bolda}^{+i}$ is
            \begin{align*}
              &\geq \big(1-(r-\epsilon/4)\big)^{2} \tag{ignoring all coordinates except $i$-th}\\
              &> \big(2r(1+\epsilon)\big)^2 \tag{from~\autoref{defn:r-epsilon-general-d} and~\autoref{eqn:always-to-be-cited}}
            \end{align*}
        \end{enumerate}

\item[(5)] \underline{$w\in \mathcal{S}_{\{\boldb,\boldb\oplus\bolde_j\}}$ for some $\boldb$ and index $j\in [d]$ such that $\{\boldb,\boldb\oplus\bolde_j\}\cap \{\bolda,\bolda\oplus\bolde_i\}=\emptyset$}:\\ %It follows Let $\bolda=(a_1,a_2,\ldots, a_d)$ and $\boldb=(b_1,b_2,\ldots, b_d)$.
    From~\autoref{table:general-d:construction}, the point $w=S_{\{\boldb,\boldb\oplus\bolde_j\}}^{\ell}$ is located at $\boldb\oplus \bolde_{j}\cdot \big(2r+\epsilon\cdot(\ell-2r)\big)$ for some $\ell\in [\delta]$ and the point $B_{\bolda}^{+i}$ is located at $\bolda \oplus \bolde_{i}\cdot (r-\epsilon/4) \oplus (\bone^d-\be_i)\cdot 2\epsilon \delta$. Since $\bolda\neq \boldb$ there exists an index $i'\in [d]$ such that $\ba[i']\neq \bb[i']$. We have the following four cases:
            \begin{enumerate}
            \item[(i)] \underline{$i'\notin \{i,j\}$} : By only considering the $i'$-th coordinate, it follows that square of the distance between $w$ and $B_{\bolda}^{+i}$ is $\geq 1>\big(2r(1+\epsilon)\big)^2$ since $r=\frac{1}{4}$ (\autoref{defn:r-epsilon-general-d}) and $\epsilon\leq 1/16$ (\autoref{eqn:always-to-be-cited}).

            \item[(ii)] \underline{$i'=i$ and $i'\neq j$} : By only considering the $i$-th coordinate, it follows that distance between $w$ and $B_{\bolda}^{+i}$ is
                    \begin{align*}
                      &\geq \big|\ba[i]+r-\epsilon/4-\bb[i]\big| \tag{only considering the $i^{\text{th}}$-coordinate}\\
                      &\geq \big|\ba[i]-\bb[i]\big| - \big|r-\epsilon/4\big| \tag{by triangle inequality}\\
                      &\geq 1 - (r-\epsilon/4) \tag{from~\autoref{defn:r-epsilon-general-d} and~\autoref{eqn:always-to-be-cited}}\\
                      &\geq 2r(1+\epsilon) \tag{from~\autoref{defn:r-epsilon-general-d} and~\autoref{eqn:always-to-be-cited}}
                    \end{align*}

            \item[(iii)] \underline{$i'=j$ and $i'\neq i$} : The  absolute value of difference between $i$-th coordinates of $w$ and $B_{\bolda}^{+i}$ is $\big|(\ba[i] + r-\epsilon/4)-\bb[i]\big|\geq (r-\epsilon/4) $ since $r=\frac{1}{4}$ and $\ba[i],\bb[i]\in [\delta]$. The  absolute value of difference between $j$-th coordinates of $w$ and $B_{\bolda}^{+i}$ is $\big|\big(\bb[j] + 2r+\epsilon(\ell-2r)\big)-\big( \ba[j] +2\epsilon\delta\big) \big|\geq \big|\ba[j]-\bb[j] \big| - \big| 2r+\epsilon(\ell-2r)-2\epsilon\delta\big| \geq \big( 1-2r- \epsilon(\ell-2r)+2\epsilon\delta\big) \geq 2r+\epsilon\delta$ where we have used the triangle inequality along with the bounds $r=\frac{1}{4}, \big|\ba[j] - \bb[j]\big|\geq 1, \ell\in [\delta]$ and $\epsilon\delta=\frac{1}{8d\delta}\leq \frac{1}{8}$. Hence, square of the distance between the points $w$ and $B_{\bolda}^{+i}$ is
                    \begin{align*}
                    &\geq \big(2r+\epsilon\delta\big)^{2} + (r-\epsilon/4)^2 \tag{only considering coordinates $i$ and $j$}\\
                    &\geq \big(2r+\epsilon\delta\big)^{2} \\
                    &\geq \big(2r(1+\epsilon)\big)^{2} \tag{from~\autoref{defn:r-epsilon-general-d} and~\autoref{eqn:always-to-be-cited}}
                    \end{align*}

            \item[(iv)] \underline{$i=i'=j$} : Since $\ba[i']\neq \bb[i']$ it follows that $|\ba[i]-\bb[i]|\geq 1$. Hence, the distance between centers of $T$ and $C_{\bolda}^{+i}$ is
                    \begin{align*}
                    &\geq \big|(\ba[i]+r-\epsilon/4)-\big(\bb[i]+2r+\epsilon\cdot (\ell-2r) \big)\big| \tag{only considering $i$-th coordinate}\\
                    &= \big| (\ba[i]-\bb[i])- \big( r+\epsilon\cdot (\ell-2r) +\epsilon/4\big) \big|\\
                    &\geq |\ba[i]-\bb[i]| - \big| r+\epsilon\cdot (\ell-2r)+\epsilon/4 \big| \tag{by triangle inequality}\\
                    &\geq 1 -\big( r+\epsilon\cdot (\ell-2r)+\epsilon/4 \big) \tag{from $|\ba[i]-\bb[i]|\geq 1, \ell\in [\delta]$ and~\autoref{defn:r-epsilon-general-d}}\\
                    &\geq 2r(1+\epsilon)  \tag{from $\ell\in [\delta]$,~\autoref{defn:r-epsilon-general-d} and~\autoref{eqn:always-to-be-cited}}
                    \end{align*}
            \end{enumerate}
\end{enumerate}
This completes the proof of the first part of the lemma. The proof of the second part of the lemma is similar, and we omit the details here.
\end{proof}

\end{document}

%% file: preamble.tex
\usepackage{ifthen}
\usepackage{comment}
\makeatletter
\usepackage[sort,comma,square,numbers]{natbib}
\renewcommand*{\NAT@spacechar}{~} %
\renewcommand\bibsection %
{
  \section*{\refname
    \@mkboth{\MakeUppercase{\refname}}{\MakeUppercase{\refname}}}
}

\usepackage{color,tikz,environ}
\usetikzlibrary{decorations.markings}
\usepackage{etoolbox}

\usepackage{titlesec}

\titleformat{\paragraph}
{\normalfont\normalsize\bfseries}{\theparagraph}{1em}{}
\titlespacing*{\paragraph}
{0pt}{3.25ex plus 1ex minus .2ex}{1.5ex plus .2ex}


\newcommand{\boundellipse}[3]%
{(#1) ellipse (#2 and #3)
}

\usepackage[margin=1in]{geometry}

\ifthenelse{\isundefined{\lipics}}
{
\usepackage{graphicx}
\usepackage[colorlinks=true,bookmarks=false]{hyperref}
\usepackage[small,bf]{caption}
\usepackage{subfigure}
\usepackage[british]{babel}
}
{}
\usepackage{xcolor}
\definecolor{darkblue}{rgb}{0,0,0.45}
\definecolor{darkred}{rgb}{0.6,0,0}
\definecolor{darkgreen}{rgb}{0.13,0.5,0}
\hypersetup{colorlinks, linkcolor=darkblue, citecolor=darkgreen,
urlcolor=darkblue}

\ifthenelse{\isundefined{\llncs}}{
  \usepackage{amsthm}
  \usepackage{authblk}
  
}

\usepackage{amsmath,amsfonts,amssymb}
\usepackage{mathtools}
\usepackage{mathrsfs}
\usepackage{wrapfig}
\usepackage{xspace}
\usepackage[shortlabels]{enumitem}
\setlist[enumerate]{nosep} %
\setlist[itemize]{nosep} %

\usepackage{multirow}

\usepackage[protrusion=true]{microtype}

\usepackage{aliascnt}
\ifthenelse{\isundefined{\llncs}}{
  \theoremstyle{plain}
}{}
\newtheorem{theorem}{Theorem}%[section]
\newaliascnt{lemma}{theorem}
\newaliascnt{corollary}{theorem}
\newaliascnt{definition}{theorem}
\newaliascnt{claim}{theorem}
\newaliascnt{proposition}{theorem}
\newaliascnt{remark}{theorem}
\newaliascnt{hypothesis}{theorem}
\newaliascnt{observation}{theorem}
\newaliascnt{conjecture}{theorem}
\newtheorem{lemma}[lemma]{Lemma}
\newtheorem{claim}[claim]{Claim}

\newtheorem{remark}[remark]{Remark}

\ifthenelse{\isundefined{\llncs}}{
  \theoremstyle{definition}
}{}
\newtheorem{definition}[definition]{Definition}
\aliascntresetthe{lemma}
\aliascntresetthe{remark}
\aliascntresetthe{corollary}
\aliascntresetthe{proposition}
\aliascntresetthe{definition}
\aliascntresetthe{claim}
\aliascntresetthe{hypothesis}
\aliascntresetthe{observation}

\newcommand{\ignore}[1]{}
\usepackage{environ,xstring}
\newif\iflabel
\newif\ifdbs
\newif\ifamp
\NewEnviron{doitall}{%
  \noexpandarg
  \expandafter\IfSubStr\expandafter{\BODY}{\label}{\labeltrue}{\labelfalse}%
  \expandafter\IfSubStr\expandafter{\BODY}{\\}{\dbstrue}{\dbsfalse}%
  \expandafter\IfSubStr\expandafter{\BODY}{&}{\amptrue}{\ampfalse}%
  \iflabel\def\doitallstar{}\else\def\doitallstar{*}\fi
  \ifdbs
    \ifamp
      \def\doitallname{align}%
    \else
      \def\doitallname{multline}%
    \fi
  \else
    \def\doitallname{equation}
  \fi
  \begingroup\edef\x{\endgroup
    \noexpand\begin{\doitallname\doitallstar}%
    \noexpand\BODY
    \noexpand\end{\doitallname\doitallstar}%
  }\x
}
\def\[#1\]{\begin{doitall}#1\end{doitall}}

\newcommand{\newreptheorem}[2]{\newtheorem*{rep@#1}{\rep@title}\newenvironment{rep#1}[1]{\def\rep@title{\bf #2 \ref*{##1}}\begin{rep@#1}}{\end{rep@#1}}}

\makeatother

\newreptheorem{theorem}{Theorem}
\newreptheorem{lemma}{Lemma}
\newreptheorem{corollary}{Corollary}

\newtheorem*{rep@thm}{\rep@title} \newcommand{\newrepthm}[2]{%
\newenvironment{rep#1}[1]{%
\def\rep@title{\autoref{##1}}%
\begin{rep@thm} }%
{\end{rep@thm} } }
\makeatother
\newrepthm{thm}{}
\newrepthm{lem}{}
\newrepthm{crl}{}

\usepackage{framed}

\usepackage{color,tikz,environ}
\usetikzlibrary{decorations.markings,arrows,plotmarks}
\tikzset{nomorepostaction/.code={\let\tikz@postactions\pgfutil@empty}}

\tikzset{middlearrow/.style={
        decoration={markings,
            mark= at position 0.5 with {\arrow{#1}} ,
        },
        postaction={decorate}
    }
}

\tikzset{onethirdarrow/.style={
        decoration={markings,
            mark= at position 0.33 with {\arrow{#1}} ,
        },
        postaction={decorate}
    }
}

\tikzset{twothirdarrow/.style={
        decoration={markings,
            mark= at position 0.67 with {\arrow{#1}} ,
        },
        postaction={decorate}
    }
}

\tikzset{endarrow/.style={
        decoration={markings,
            mark= at position 0.9 with {\arrow{#1}} ,
        },
        postaction={decorate}
    }
}

\tikzset{startarrow/.style={
        decoration={markings,
            mark= at position 0.1 with {\arrow{#1}} ,
        },
        postaction={decorate}
    }
}